\documentclass{article}
\usepackage{amssymb,amsmath,amsthm,graphicx,cite,yhmath}
\allowdisplaybreaks[3]

\makeatletter
\renewcommand*{\@fnsymbol}[1]{\ensuremath{\ifcase#1\or *\or 1\or 2\or
   3\or 4\or 5\or 6\or 7 \or 8 \else\@ctrerr\fi}}
\makeatother

\numberwithin{equation}{section}
\newtheorem{theorem}{Theorem}[section]

\newtheorem{lemma}[theorem]{Lemma}

\theoremstyle{definition}

\theoremstyle{remark}
\newtheorem{remark}[theorem]{Remark}

\newcommand{\R}{{\mathbb{R}}}

\DeclareMathOperator{\esssup}{{\mathrm{ess\;sup}}}

\begin{document}

\title{Density behavior of spatial birth-and-death stochastic
evolution of mutating genotypes under selection rates
\thanks{The work is partially supported by SFB project 701, RFFI 11-01-00485a, NSF Grant DMS-1008132, and ZiF Research Group ``Stochastic Dynamics: Mathematical
Theory and Applications'', Bielefeld, Germany}}

\author{Dmitri Finkelshtein\thanks{%
        Fakult\"{a}t f\"{u}r Mathematik, Universit\"{a}t Bielefeld,
        Postfach 110 131, 33501 Bielefeld, Germany,
        e-mail: finkelst@math.uni-bielefeld.de} 
        \and Yuri Kondratiev\thanks{%
        Fakult\"at f\"ur Mathematik, Universit\"at Bielefeld, 
        Postfach 110 131, 33501 Bielefeld, Germany,
        e-mail: kondrat@math.uni-bielefeld.de} 
        \and Oleksandr Kutoviy\thanks{%
        Department of Mathematics, MIT,
        77 Massachusetts Avenue 2-155,
        Cambridge, MA, USA,
        e-mail: kutovyi@mit.edu;
        Fakult\"at f\"ur Mathematik,
        Universit\"at Bielefeld,
        Postfach 110 131, 33501 Bielefeld, Germany,
        e-mail: kutoviy@math.uni-bielefeld.de}
        \and Stanislav Molchanov\thanks{%
        Department of Mathematics and Statistics, University of North Carolina at Charlotte, 
        Charlotte, NC 28223, USA,
        e-mail: smolchan@uncc.edu} 
        \and Elena Zhizhina\thanks{%
        Institute for Information Transmission Problems 
        of the Russian Academy of Sciences, Moscow, Russia,
        e-mail: ejj@iitp.ru}%
        }

\date{\today}
\maketitle

\begin{abstract}
We consider birth-and-death stochastic evolution of genotypes with different lengths. The genotypes might mutate that provides a stochastic changing of lengthes by a free diffusion law. The birth and death rates are length dependent which corresponds to a selection effect. We study an asymptotic behavior of a density for an infinite collection of genotypes. The cases of space homogeneous and space heterogeneous densities are considered.


\end{abstract}

\section{Description of model}

We start with a heuristic discussion of a model, describing spatial
evolution of mutating genotypes under selection rates. Each genotype
might be characterized by a pair $\widehat{x}:=(x,s_x)$. Here $x\in\R^d$
is a location in the Eucledian space occupied by this genotype, and
a mark $s_x$ is its quantitative characteristic. We will consider,
cf. \cite{CK,K}, a continuous-gene space model. Namely, $s_x\in
\R_+:=[0,+\infty)$ will be understood as a length of a genotype
located at~site~$x$.

We describe an infinite collection of genotypes as a configuration
$\widehat{\gamma}:=\{\widehat{x}\}$. Having in mind that in the reality any
individual with a given genotype has not only position in space but
also non-zero size, we assume that $\gamma:=\{x\}$ is a locally
finite subset in $\R^d$. Namely, $\gamma\cap\Lambda$ is a finite set
for any compact $\Lambda\subset\R^d$. Let $\Gamma$ and
$\widehat{\Gamma}$ be the spaces of such $\gamma$'s and
$\widehat{\gamma}$'s, accordingly.

In the present paper, we deal with mutating genotypes. Omitting the
nature of these mutations, we suppose that they lead to a stochastic
evolution of marks $s_x$, given by Brownian motion on $\R_+$ with
absorption at $0$. We consider a birth-and-death stochastic dynamics
of mutating genotypes. It means that at any random moment of time
the existing genotype may disappear (die) from the configuration or
may produce a new one. This new genotype will be placed at other
location in the space. It has the parent's genotype at the moment of
birth, but then it immediately involves in a mutation process. This
may be understood as an expansion of genotypes along the space. The
probabilistic rates of birth and death of a genotype are independent
of the rest of configuration, however, we suppose that they depend
on sizes of genotypes. In fact, it means that we have selection in
rates of birth and death. It is natural for biological systems that
genotypes with very shot as well as very long length have less
possibilities for surviving and reproduction, see e.\,g. \cite{Baake, Burger}.

The heuristic Markov generator of the dynamics described above may be
given by
\begin{align}
(L F)(\widehat{\gamma}) &=  \sum_{x\in \gamma} b(s_x) \int_{R^d} a(x-y)
\bigl(F(\widehat{\gamma} \cup \{ y, s_x \}) - F(\widehat{\gamma})\bigr) dy
\notag
\\&\quad + \sum_{x\in \gamma}  d(s_x) \bigl(F(\widehat{\gamma}
\setminus \{ x, s_x \}) - F(\widehat{\gamma})\bigr)  +  \sum_{x \in
\gamma} \frac{\partial^2}{\partial s_x^2}  F(\widehat{\gamma}).
\label{L}
\end{align}

The first term in \eqref{L} describes the birth of genotypes. This
reproduction step involves selection as well as expansion of
genotypes along the space. The function $a$ describes an expansion
(migration) rate, it is independent on marks $s_x, s_y$. Function
$b$ is associated with stabilizing selection. It prescribes that
some lengths may be ranked against the other lengths. Genotypes with
optimal (or at least more optimal) lengths are assumed to breed and
to spread more intensively. We assume that $0\leq a\in L^1(\R^d)$,
$a$ is an even non-negative function, $b:\R_+\to \R_+$ with
$b(0)=0$. Without loss of generality we suppose that
$\int_{\R^d}a(x)dx=1$.

The second term in \eqref{L} corresponds to the death of genotypes. We assume here
that the death rate $d:\R_+\to \R_+$ depends only on a length of a
genotype, and does not depend on a location of genotype in the
space. The shape of $d$ will be discussed below.

The third term describes mainly mutations of genotypes, but also can
include all random changes within the genotype, such as:
duplication, genetic drift, etc. This differential operator is a
modification of the generator for a random jump mutation model on
the continuous space. Let us note that the third term is the direct
sum of operators. That means that we assume that each offspring
develops independently on others and we do not consider any
interaction between existing genotypes.

Note that models of this type (without expansion), so-called
mutation-selection models, play an important role in analysis of
many problems of population genetics, see e.\,g. \cite{Baake, Burger}.

To give a rigorous meaning to the expression \eqref{L} we consider
the following classes of functions. Let $\mathcal{D}$ consist of all
functions $\varphi:\R^d\times\R_+\to \R$ which have bounded support
in $\R^d\times(0,\infty)$, and $\varphi$ is a continuous functions
in the first variable and twice continuously differentiable in the
second variable. For any $\varphi\in\mathcal{D}$ the following
expression is well-defined:
\[
\langle\varphi,\widehat{\gamma}\rangle:=\sum_{x\in\gamma}\varphi(x,s_x),
\]
since the summation will only be taken over the finite set $\gamma_{\Lambda}:=\gamma\cap\Lambda$
for some compact $\Lambda\subset\R^d$. Let
$\varphi_1,\ldots,\varphi_N\in\mathcal{D}$ and $f:\R^N\to \R$ be
twice continuously differentiable function on $\R^N$ bounded
together with all its partial derivatives. The class of all
functions of the form
$$F(\widehat{\gamma})=f(\langle\varphi_1,\widehat{\gamma}\rangle,\ldots,
\langle\varphi_N,\widehat{\gamma}\rangle),\quad \widehat{\gamma}\in\widehat{\Gamma}$$
we denote by $\mathcal{F}$. It is worth noting that  for any $F\in\mathcal{F}$ the value of $F(\widehat{\gamma})$ does not depend on
on those $\widehat{x}\in\widehat{\gamma}$ which are outside of the union of
supports of $\varphi_1,\ldots,\varphi_N$. In particular, the summation in the second term of
\eqref{L} will only be taken over a finite subset of each $\gamma$, hence this term is well-defined. Analogously, for each $x$
which is outside of the union of supports above,
$\dfrac{\partial^2 F}{\partial s_x^2} (\widehat{\gamma})=0$. Similarly,
the integral in the first term of \eqref{L} will be taken over a
compact set. Moreover, if, additionally, $a$ has compact support in
$\R^d$ the sum before integral will be also finite. For a general
integrable function $a$, this sum is a series which may converges
only a.s. in the following sense.

Let $\mu$ be a probability measure (state) on the space
$\widehat{\Gamma}$ with $\sigma$-algebra described e.\,g. in \cite{Ka}. A
function $k_\mu:\R^d\times\R_+\to \R$ is called a density (or a
first correlation function) of the measure $\mu$ if for any
$\varphi:\R^d\times\R_+\to \R$ such that $\varphi\cdot k_\mu\in
L^1(\R^d\times\R_+)$ we have:  the function $\langle\varphi,\cdot\rangle$
belongs to $L^1(\widehat{\Gamma},\mu)$ and
\[
\int_{\widehat{\Gamma}}\langle\varphi,\widehat{\gamma}\rangle
d\mu(\widehat{\gamma})= \int_{\R^d}\int_{\R_+}\varphi(x,s)
k_\mu(x,s)dxds.
\]
In this case, $\langle\varphi,\widehat{\gamma}\rangle$ is well-defined
for $\mu$-almost all $\widehat{\gamma}\in\widehat{\Gamma}$.

It is obvious, that for $a\in L^1(\R^d)$ and any probability measure
$\mu$ on $\widehat{\Gamma}$ with the bounded density $k_\mu$ the first
term in \eqref{L} is well-defined for $\mu$-almost all
$\widehat{\gamma}\in\widehat{\Gamma}$ and $F\in\mathcal{F}$.

The construction of evolution of states with the generator given by
\eqref{L} is usually related with the construction and properties of
evolution of densities and higher-order correlation functions (see
e.\,g. \cite{KKP} for the case without marks). The aim of the present
paper is to study the evolution of the density only. Therefore, we
suppose that there exists an evolution of measures given by
\begin{equation}\label{FP}
\frac{d}{dt}\int_{\widehat{\Gamma}} F d\mu_t=\int_{\widehat{\Gamma}} LF
d\mu_t, \quad F\in\mathcal{F}
\end{equation}
with initial measure $\mu_0$ at $t=0$. We assume also that $k_t$
be a density of $\mu_t$. Then, for
$F_\varphi(\widehat{\gamma}):=\langle\varphi,\widehat{\gamma}\rangle$,
$\varphi\in\mathcal{D}$ we obtain
\[
(LF_\varphi)(\widehat{\gamma}) = \sum_{x\in \gamma} \int_{\R^d} a(x-y)
b(s_x) \varphi(y,s_x) dy  - \sum_{x\in \gamma}  d(s_x)
\varphi(x,s_x) + \sum_{x \in \gamma} \frac{\partial^2}{\partial
s_x^2} \varphi(x,s_x).
\]
Therefore,
\begin{align}
\int_{\widehat{\Gamma}} (LF_\varphi)(\widehat{\gamma})
d\mu_t(\widehat{\gamma})&= \int_{\R^d} \int_{\R_+} k_t(x,s) \int_{\R^d}
a(x-y) b(s) \varphi(y,s) dy dsdx\notag
\\&\quad - \int_{\R^d} \int_{\R_+} k_t(x,s) d(s) \varphi(x,s) dsdx\notag\\&\quad+
\int_{\R^d} \int_{\R_+} k_t(x,s) \frac{\partial^2}{\partial s^2}
\varphi(x,s)dsdx.\label{eqn1}
\end{align}
On the other hand,
\begin{equation}\label{eqn2}
\frac{d}{dt}\int_{\widehat{\Gamma}} F_\varphi(\widehat{\gamma})
d\mu_t(\widehat{\gamma})=\frac{d}{dt}\int_{\R^d} \int_{\R_+}
k_t(x,s)\varphi(x,s) dsdx.
\end{equation}
Since $\varphi\in\mathcal{D}$ is arbitrary, by \eqref{FP}, \eqref{eqn1}, \eqref{eqn2}, the
densities $k_t$ should satisfy (in a weak sense) the following
differential equation
\begin{equation}\label{CP}
\frac{\partial}{\partial t}
k_t(x,s)=b(s)\int_{\R^d}a(x-y)k_t(y,s)dy-d(s)
k_t(x,s)+\frac{\partial^2}{\partial s^2}k_t(x,s).
\end{equation}
Using the assumption $\int_{\R^d}a(x)dx=1$ we may rewrite \eqref{CP}
as follows
\begin{align}\label{maineqn}
\frac{\partial}{\partial t} k_t(x,s)&=(\mathsf{A}k_t)(x,s)-(\mathsf{H} k_t)(x,s),\\
(\mathsf{A}k_t)(x,s)&:=b(s)\int_{\R^d}a(x-y)\bigl(k_t(y,s)-k_t(x,s)\bigr)dy,\label{A}\\
(\mathsf{H} k_t)(x,s)&:=-\frac{\partial^2}{\partial
s^2}k_t(x,s)+(d(s)-b(s))k_t(x,s).\label{H1}
\end{align}
It is worth noting that appearance of effective potential $v(s) =
d(s) - b(s)$ is inspired by the evolution mechanism of the spatial
microscopic model. The function $v(s)$ has meaning of a fitness
function, see e.\,g. \cite{Burger}. The typical graphs of $b(s)$
and $d(s)$ motivated by the biological applications are given on the Figures~\ref{fig1} and \ref{fig2}, correspondingly.
\begin{figure}[h]
  \centering \includegraphics{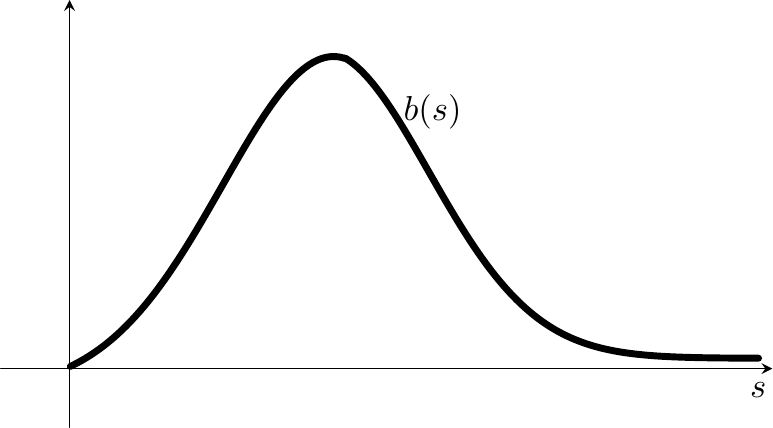}\\
  \centering \caption{Fast decreasing of $b(s)$ if $s\to\infty$}\label{fig1}
\end{figure}
\begin{figure}[h]
  \centering \includegraphics{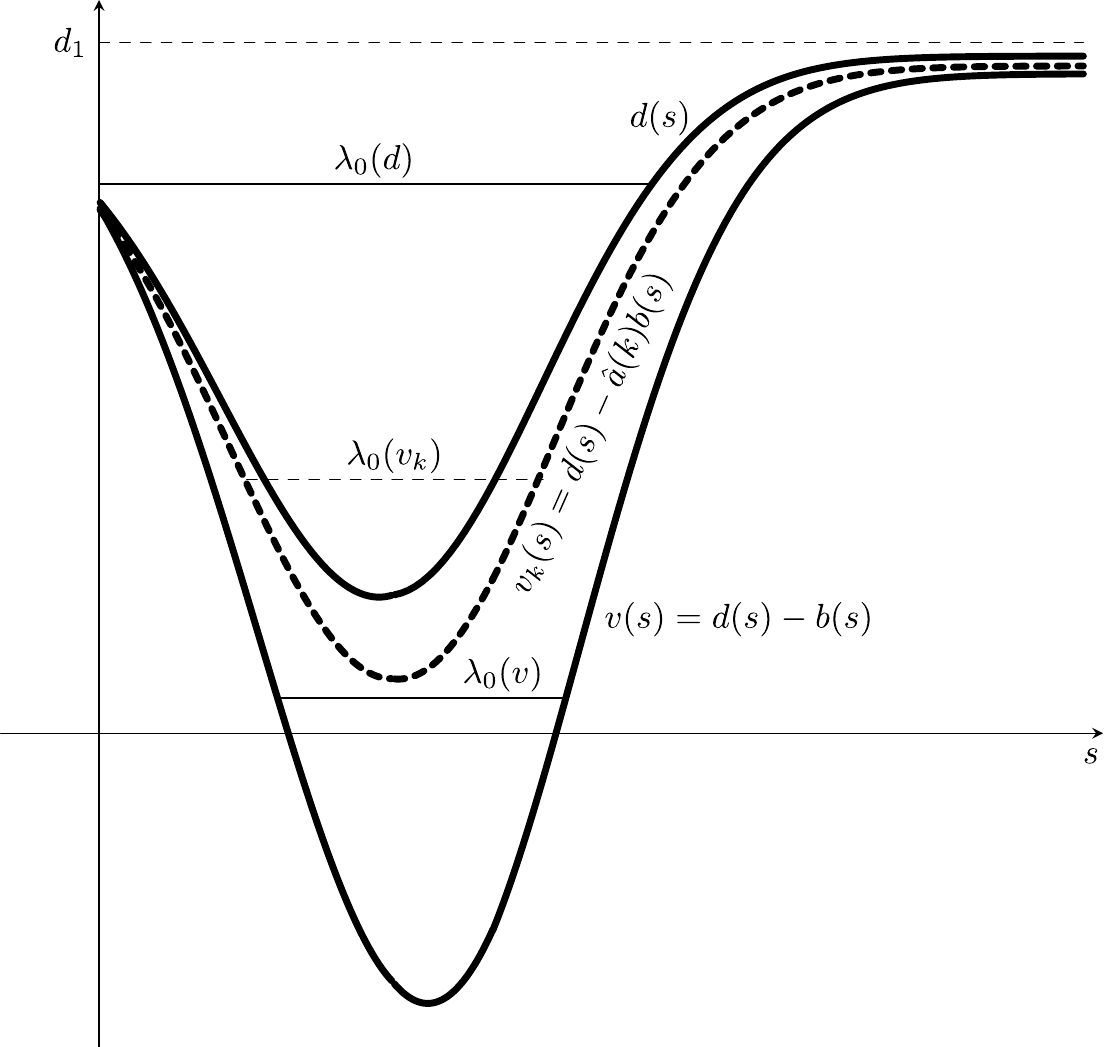}\\
  \centering \caption{If $s \to \infty$ then either $d(s) \to
+\infty$ or $d(s) \to d_1$, $d(s)\leq d_{1}$, $s\in\R_{+}$ and
$d(s_0)< b(s_0)$}\label{fig2}
\end{figure}

In the next sections we will study the classical solution of
\eqref{maineqn}--\eqref{H1} with initial conditions $k_0$ in different
Banach spaces.

\section{Asymptotic of a spatially uniform density }

Let $\mathcal{H}=L^2(\R_+)$ be a real Hilbert space. Let us define the following class of functions
$$
X=\biggl\{\sum_{i=1}^{n}c_{i}\psi_{i}(x)\varphi_{i}(s)\biggm| c_{i}\in\R,\;\psi_{i}\in L^{\infty}(\R^{d}),\;\varphi_{i}\in\mathcal{H}, \;i=1,\ldots,n,\;n\in\mathbb{N}\biggr\}.
$$
By  $\mathcal{X}$ we denote the closure of $X$ with respect to the norm
\[
\|k\|_\mathcal{X}:=\esssup_{x\in\R^d}\|k(x,\cdot)\|_{\mathcal{H}}.
\]
Hence one can naturally embed $\mathcal{H}$ into $\mathcal{X}$ as
set of functions which are constants in $x$. We will use the same
notations for function $f\in\mathcal{H}$ as an element of
$\mathcal{X}$.

Suppose that there exists $\omega\geq0$ such that
\begin{equation}\label{bbb}
v(s):=d(s)-b(s)\geq -\omega, \quad s\in\R_+.
\end{equation}
Let $C_0^\infty(\R_+)$ consist of all smooth functions $f$ on $\R_+$
with bounded support such that $f(0)=0$. Then the operator
\[
(Hf)(s):=-\frac{d^2 f(s)}{d^2 s} + v(s)f(s)
\]
with a domain
$C_0^\infty(\R_+)$ is essentially self-adjoint in $\mathcal{H}$
(see e.\,g. \cite{BSh}). Let $\bigl(\bar{H},
\mathrm{Dom}(\bar{H})\bigr)$ be its self-adjoint closure in
$\mathcal{H}$. Let $\mathsf{D}\subset\mathcal{X}$ consist of all
functions $k\in\mathcal{X}$ such that, for a.a. $x\in\R^d$,
$k(x,\cdot)\in\mathrm{Dom}(\bar{H})$.
\begin{lemma}\label{le}
Let \eqref{bbb} hold and $b\in L^\infty(\R_+)$. Then
$(\mathsf{A}-\mathsf{H},\mathsf{D})$ is a generator of a
$C_0$-semigroup $S(t)$ in $\mathcal{X}$.
\end{lemma}
\begin{proof}
Since $v$ is bounded from below, $(-\bar{H}f,f)_\mathcal{H}\leq
\omega\|f\|^2_\mathcal{H}$ for any
$f\in\mathrm{Dom}(\bar{H})$. Therefore, by e.\,g. \cite[Example
II.3.27]{EN}, $\bigl(-\bar{H}, \mathrm{Dom}(\bar{H})\bigr)$ is a
generator of a $C_0$-semigroup $T_{\bar{H}}(t)$ in $\mathcal{H}$,
and moreover, $\|T_{\bar{H}}(t)\|\leq e^{t\omega}$, $t\geq0$. Then,
by a version of Hille--Yosida theorem (see e.\,g. \cite[Corollary
II.3.6]{EN}), for each $\lambda > \omega$,
$\lambda\in\rho(-\bar{H})$ and $\|R(\lambda, -\bar{H})\|\leq
(\lambda-\omega)^{-1}$. Here and below $\rho(B)$ and $R(\lambda, B)$
denotes a resolvent set and a resolvent of a closed operator $B$,
correspondingly. By \eqref{H1} and the properties of $\bar{H}$, it
is evident that $(-\mathsf{H}, \mathsf{D})$ is a closed densely
defined operator in $\mathcal{X}$. Moreover, $\rho(-\mathsf{H})
=\rho(-\bar{H})$, and, for each $\lambda\in\rho(-\mathsf{H})$,
\[
\bigl(R(\lambda, -\bar{H})k(x,\cdot)\bigr)(s)=\bigl(R(\lambda,
-\mathsf{H})k\bigr)(x,s), \quad k\in\mathcal{X}, x\in\R^d, s\in\R_+.
\]
As a result,
\begin{align*}
\|R(\lambda,
-\mathsf{H})k\|_\mathcal{X}&=\esssup_{x\in\R^d}\bigl\|\bigl(R(\lambda,
-\mathsf{H})k\bigr)(x,\cdot)\bigr\|_\mathcal{H}\\&=\esssup_{x\in\R^d}\bigl\|R(\lambda,
-\bar{H})k(x,\cdot)\bigr\|_\mathcal{H}\\&\leq
(\lambda-\omega)^{-1}\esssup_{x\in\R^d}\|k(x,\cdot)\|_\mathcal{H}=(\lambda-\omega)^{-1}\|k\|_\mathcal{X}.
\end{align*}
Hence, by the version of Hille--Yosida theorem mentioned above,
$(-\mathsf{H}, \mathsf{D})$ is a generator of a $C_0$-semigroup
$T_{\mathsf{H}}(t)$ in the space $\mathcal{X}$, and moreover,
$\|T_{\mathsf{H}}(t)\|\leq e^{t\omega}$, $t\geq0$. Next, since $b\in L^\infty(\R_+)$, we have, for any
$k\in\mathcal{X}$ and for a.a. $x\in\R^d$,
\begin{align*}
&\|(\mathsf{A}k)(x,\cdot)\|_\mathcal{H}\\&\leq
\|b\|_{L^\infty(\R_+)}\biggl( \int_{\R_+} \biggl( \int_{\R^d}
a(x-y)\bigl(k(y,s)-k(x,s)\bigr)dy\biggr)^2
ds\biggr)^{\frac{1}{2}}\\&\leq \|b\|_{L^\infty(\R_+)}\biggl(
\int_{\R_+} \int_{\R^d} a(x-y)\bigl|k(y,s)-k(x,s)\bigr|^2dy
ds\biggr)^{\frac{1}{2}}\\&\leq
\sqrt{2}\,\|b\|_{L^\infty(\R_+)}\biggl( \int_{\R_+} \int_{\R^d}
a(x-y)|k(y,s)|^2 dy ds+\int_{\R_+}|k(x,s)|^2 ds\biggr)^{\frac{1}{2}}
\\&\leq 2 \|b\|_{L^\infty(\R_+)}\Bigl(
\esssup_{x\in\R^d}\|k(x,\cdot)\|_\mathcal{H}^2\Bigr)^{\frac{1}{2}}\leq
2 \|b\|_{L^\infty(\R_+)}\|k\|_\mathcal{X}.
\end{align*}
Therefore, $\mathsf{A}$ is a bounded operator in $\mathcal{X}$ with
$\|\mathsf{A}\|\leq 2 \|b\|_{L^\infty(\R_+)}$. Then, by e.\,g.
\cite[Theorem III.1.3]{EN}, the operator $-\mathsf{H}+\mathsf{A}$
with domain $\mathsf{D}$ generates a $C_0$-semigroup $S(t)$ in
$\mathcal{X}$. Moreover,
\[
\|S(t)\|\leq \exp\{(\omega+2 \|b\|_{L^\infty(\R_+)})t\}, \quad
t\geq0.
\]
\end{proof}

Our goal is to study the asymptotic behavior of $k_t
(x,s)=S(t)k_0(x,s)$ as $t\to\infty$. Here $k_0\in\mathcal{X}$. In
the particular case $k_0(s)\in\mathcal{H}$, one can solve this
problem in details.

\begin{theorem}
Let \eqref{bbb} hold and $b\in L^\infty(\R_+), \; b(s) \ge 0$.
Suppose additionally that the operator $\bar{H}$ in $\mathcal{H}$
has either simple discrete spectrum $\lambda_0< \lambda_1<\ldots < \lambda_n \to \infty $, as $n \to \infty$, or continuous spectrum
$[\lambda,+\infty)$ and a finite number of simple eigenvalues
$\lambda_0<\lambda_1<\ldots \; <\lambda_n<\lambda$. Consider the
initial condition given by $k_0 (x,s) = \varrho_0 (s)$, for a.a.
$x\in\R^d$, $s\in\R_+$, where $\varrho_0\in \mathcal{H}$. Then
\begin{equation}\label{asympt}
\bigl\|S(t)k_0(s) - e^{-t\lambda_0} c_0 \psi_0(s) \bigr\|_\mathcal{X} = O
\bigl(e^{-t\lambda_1}\bigr), \quad t\to \infty,
\end{equation}
where $\psi_0(s)$ is the eigenfunction of the operator $\bar{H}$
corresponding to the eigenvalue $\lambda_0$, $c_0 =(\varrho_0,
\psi_0)_\mathcal{H}$, and $\lambda_1 > \lambda_0$.
\end{theorem}

\begin{proof}
By the proof of Lemma~\ref{le}, the operator $(-\mathsf{H},
\mathsf{D})$ is a generator of a $C_0$-semigroup $T_{\mathsf{H}}(t)$
in $\mathcal{X}$ and $\mathsf{A}$ is a bounded operator in
$\mathcal{X}$. Then, by the Trotter formula (see e.\,g. \cite[Exersise
III.5.11]{EN}), we have
\[
S(t)k_0=\lim_{n\to\infty}\Bigl(T_\mathsf{H}\Bigl(\frac{t}{n}\Bigr)e^{\frac{t}{n}\mathsf{A}}\Bigr)^nk_0,
\]
where the limit is considered in the sense of norm in $\mathcal{X}$.
Note that for any $f\in\mathcal{H}\subset\mathcal{X}$,
$\mathsf{A}f=0$, therefore, $e^{t\mathsf{A}}f=f$ for all $t>0$.
Since $k_0$ does not depend on $x$, we have that
$T_\mathsf{H}\bigl(\frac{t}{n}\bigr)e^{\frac{t}{n}\mathsf{A}}k_0=T_\mathsf{H}\bigl(\frac{t}{n}\bigr)k_0$,
and the latter function does not depend on $x$ also. As a result,
\[
 S(t)k_0=T_\mathsf{H}(t)k_0=T_{\bar{H}}(t)\varrho_0.
\]
Therefore, it is enough to show that
 \[
 \|T_{\bar{H}}(t)\varrho_0- e^{-t\lambda_0}c_0 \psi_0\|_\mathcal{H}
 =O\bigl(e^{-t\lambda_1}\bigr), \quad \lambda_1 > \lambda_0, \quad t\to\infty.
 \]
The latter asymptotic follows from the general spectral theory of
self-adjoint operators, see e.\,g. \cite{RSN}. Using spectral
decomposition of the self-adjoint operator $\bar{H}$ in the Hilbert
space $\mathcal{H}$, we get
$$
T_{\bar{H}} (t) \varrho_0 =  \int_{\sigma{(\bar{H})}} e^{-tu}
dE_{\bar{H}} (u) \varrho_0,
$$
where $E_{\bar{H}}$ is the spectral measure of $\bar{H}$ and the
integral is taken over the spectrum of $\bar{H}$. Then,
$$
 \|T_{\bar{H}}(t)\varrho_0- e^{-t\lambda_0}c_0 \psi_0\|^2_\mathcal{H}
\leq e^{-2t \lambda_1} \| P_{{\cal H}'} \varrho_0 \|^2_\mathcal{H} =
O(e^{-2t \lambda_1}),
$$
where $P_{{\cal H}'}$ is the projection on $\mathcal{H}' :=
\mathcal{H} \ominus \{\psi_0\}$. (Note that $\lambda_1$ may be equal
to $\lambda$.) The statement is proved.
\end{proof}


\begin{remark}
Asymptotic formula \eqref{asympt} describes, in
particular,
an asymptotical shape of the density. Assume that the
initial
density has the form $k_0(s, x)= \varrho (s)$,
$x\in\R^d, \varrho \in \mathcal{H}$,
which is uniform over the space $\R^d$ but dependent on
mark $s$,
where the dependence is defined by an arbitrary function
$\varrho (s) \in \mathcal{H}$. Then on a large-time scale
(when
$t$ is large enough) we get a density $k_t(s,x)$, which is
again
uniform over the space: $k_t(s, x)= \varrho_t (s)$, and
function $\varrho_t (s)$ specifying the mark dependence in
the
density $k_t (s, x)$ has now a definite shape. It is
shaped like
the first eigenfunction $\psi_0(s)$ of the operator $H$.
That means that an optimal range of mark values appears
under
the long-time evolution.
\end{remark}

\begin{remark}\label{rem2}
Consider the basic Schr\"odinger operator in $L^2(\R_+)$ with
absorption boundary condition:
\begin{equation}\label{BSO}
H f = - \frac{d^2 f}{d s^2} + v\cdot f, \quad v \ge -\omega, \quad
f(0) =0. 
\end{equation}
Then the behavior of the populations in whole depends on the sign of
the minimal eigenvalue $\lambda_0$ of the operator $H$: if
$\lambda_0>0$, then populations are vanishing, if $\lambda_0 <0$,
then populations are increasing. The case $\lambda_0 =0$
(``equilibrium'' regime) is of particular interest. As follows from
the well-known facts on spectrum of one-dimensional Schr\"{o}dinger
operator, see e.\,g. \cite{BSh}, the sign of $\lambda_0$ depends on
the shape of the function $v(s)= d(s) - b(s)$. Let us distinguish
two interesting cases.
\begin{enumerate}
  \item Let $v(s)\to +\infty$, $s \to +\infty$, that means that $d(s)\to+\infty$;
in this case the spectrum of $\bar{H}$ is discrete and simple;
moreover, if $v(s) \ge 0$ then $\lambda_0 >0$;
  \item in the case $d(s)\leq d_{1}$, $s\in\R_{+}$ and $d(s) \to d_1 >0, \; b(s) \to 0$,
$s\to\infty$ the spectrum has a continuous component $[d_1, +
\infty)$ and possibly a discrete set of simple eigenvalues which are
smaller than $d_1$. A simple sufficient condition for the existence
of the ground state in this case has the following form: one can
insert rectangle with the sites $h,l>0$ such that $\sqrt{h}l>\pi/2$
between level $d_{1}$ and graph of $v(s)$ (see Figure 3).
\begin{figure}[ht]
  \centering \includegraphics{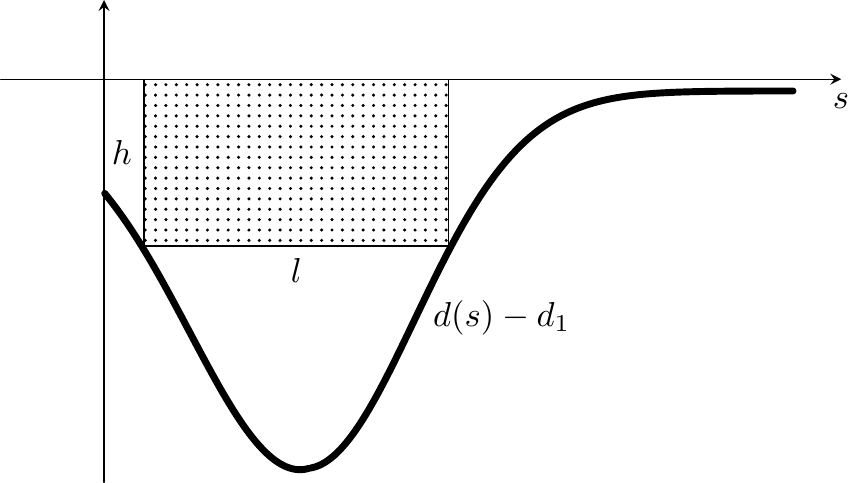}\\
  \caption{Sufficient condition for the existence
of a ground state when $d(s)\leq d_{1}$, $d(s)\to d_1$, $b(s)\to0$ as $s\to\infty$}\label{fig3}
\end{figure}

\end{enumerate}
\end{remark}

\begin{remark}
The behavior of the populations in whole depends on the sign of the
minimal eigenvalue $\lambda_0$ of the operator $H$: if
$\lambda_0>0$, then populations are vanishing, if $\lambda_0 <0$,
then populations are increasing in the following sense. Denote by
$n(D) = |\widehat \gamma \cap D|$ a random variable equal to the number
of particles from the configuration lying inside of the finite
volume $ D = D_x \times D_s, \ D_x \subset \R^{d}, \ D_s \subset
\R_{+}$. Then
$$
n(D) = \sum_{\widehat x \in \widehat \gamma} \chi_D (\widehat x) = \langle \chi_D, \widehat \gamma \rangle,
$$
where $\chi_D$ is the characteristic function of $D$. The asymptotic
formula (\ref{asympt}) immediately implies the following asymptotic
for the average of $n(D)$ as $t \to \infty$:
\begin{align*}
\mathbb{E}_t  n(D)  &=  \int_{\widehat \Gamma} \langle \chi_D, \ \widehat
\gamma \rangle \ d \mu_t (\widehat \gamma)  =  \int_{\R^d} \int_{\R_+}
\chi_D (x,s) k_t (x,s) dx  ds \\
&= \  e^{-t \lambda_0} c_0 |D_x| \int_{D_s} \psi_0 (s) ds \ (1  +
o(1) )
\end{align*}
\end{remark}

\section{Asymptotic of a spatially local density}

In this section we consider the long time behavior of the density
$k_t(s,x)$ assuming that $k_t (s,x) \in L^2 (\R_+, \R^d)$. Thus, we
consider the Cauchy problem:
\begin{align}\label{PP}
\frac{\partial \psi}{\partial t} &= {\cal L} \psi, \quad \psi =
\psi_t(s,x), \quad s\in \R_+,\; x \in \R^d, \; t \ge 0,\\
({\cal L} \psi)(s,x) &= \frac{\partial^2 \psi(s,x)}{\partial s^2} -
v(s)\psi(s,x)  b(s) \int_{\R^d} a(x-y)(\psi(s,y)- \psi(s,x)) dy,\notag\\
\psi_0(s,x) &\ge 0, \quad \psi_0 \in C_0 (\R_+, \R^d),\notag\\ \psi_t(0,x)
&= 0 \; (\mbox{absorption boundary condition (ABC)}).\notag
\end{align}
Here  $a(x-y)$ and $v(s) = d(s) - b(s)$ are the same functions as in
the preceding section.

\begin{theorem}
Let us assume that operator $H$ (\ref{BSO}) has ground state
$\lambda_0$, and additionally that
$$
a(z)\sim\frac{c_{0}}{|z|^{d+\alpha}}, \quad 0<\alpha\leq  2, \quad
\mbox{as} \quad z \to \infty.
$$
Then the solution of the parabolic problem (\ref{PP}) has the
following asymptotic as $t \to \infty$:
\begin{equation}\label{asymptPP}
\psi_t (s,x) = \frac{e^{-t\lambda_0}}{t^{d/\alpha}} C_d(\alpha)
\ (\psi_0 (s), \widehat u_0 (s,0)) \ \psi_0(s) \  p_\alpha\Bigl(
\frac{x}{(\tilde c_0 t)^{1/\alpha}} \Bigr) (1+o(1)),
\end{equation}
where $\psi_0 (s)$ is the eigenfunction of the operator $H$
corresponding to the eigenvalue $\lambda_0$, $C_d(\alpha), \tilde
c_0$ are positive constants depending on functions $b(s), a(x-y)$;
$\; p_\alpha(\cdot)$ is a density of a d-dimensional symmetric
$\alpha$-stable distribution.
\end{theorem}

Note that for smaller $0< \alpha \leq 2$ the pre-exponential term is
decreasing faster.

\begin{remark}
In the case when $|x| \ll t^{1/\alpha}$, the asymptotic of
(\ref{asymptPP}) can be written as follows:
\begin{equation}\label{asymptG}
\psi_t (s,x) =  C_d (\alpha) \ V(\alpha) \frac{e^{- t \lambda_0
}}{t^{d/\alpha}} \psi_0 (s) (\psi_0 (s), \widehat u_0 (s,0))(1 + o(1)),
\end{equation}
where $V(\alpha)\ = \ \frac{1}{(2 \pi)^d}  \int_{\R^d} e^{-
|q|^\alpha} d q $.
\end{remark}
\begin{proof}
Since the operator
$$
(B \psi) (s,x) = b(s) \int_{\R^d} a(x-y)(\psi(s,y)- \psi(s,x))
dy
$$
is bounded and  self-adjoint in $L^2(\R_+, \R^d)$ as well as the
operator $H$ is essentially self-adjoint in $L^2(\R_+)$, we conclude
that the operator ${\cal L}$ is self-adjoint and bounded from above
in $L^2(\R_+, \R^d)$. Consequently, the operator ${\cal L}$
generates a $C_0$-semigroup in $L^2(\R_+, \R^d)$: $\; \psi_t(s,x) = e^{t {\cal L}} \psi_0(s,x)$.

We will construct now the spectral representation of $\cal L$ as the
direct integral of 1-D Schr\"odinger operators $H_k, \ k \in \widehat
\R^d$ using the Fourier transform over $x \in \R^d$. Consider
$\psi(s,x) \in L^2(\R_+, \R^d)$ and present it using the Fourier
duality in the form
$$
\psi(s,x) = \frac{1}{(2 \pi)^d} \int_{\widehat \R^d} e^{-i (k,x)}
\widehat \psi (s,k) dk
$$
Then the operator $\widehat{\cal L}$ in $L^2 (\R_+ \times \widehat \R^d)$
has the following form
$$
(\widehat{\cal L} \widehat\psi)(s,k) = \frac{\partial^2 \widehat
\psi(s,k)}{\partial s^2} - v(s) \widehat \psi(s,k) - b(s)(1 - \widehat a(k)) \widehat
\psi(s,k), \quad \quad \widehat\psi(0,k)=0.
$$
Here
$$
\widehat a(k) = \int_{\R^d} e^{i(k,z)} a(z) dz,
$$
and $\widehat a(0) = \int_{\R^d} a(z) dz =1$ by the normalization condition.

Let us introduce for each $k \in \widehat \R^d$ the Schr\"odinger
operator $H_k:=H + B_k$  where the operators $H$ and $B_k$ act in $L^2(\R_+)$ as follow
\begin{gather*}
H \psi (s) := -\frac{d^2 \psi}{d s^2} + v(s) \psi(s), \\ B_k \psi(s):=  \ b(s)(1-\widehat
a(k)) \psi (s), \\
\psi(0)=0,  \quad \psi\in C_0(\R_+).
\end{gather*}
One can rewrite
$$
H_k  \psi(s) = - \frac{d^2\psi(s) }{d s^2} + v_k(s)
\psi(s),
$$
where
$$
v_k(s) := v(s) + b(s) (1 -
\widehat a(k)) = d(s) - b(s)\widehat a(k).
$$
It is worth noting that
\begin{equation}\label{A1}
\widehat a(k) < \widehat a(0)= \int_{\R^d} a(z) dz = 1, \  k
\in\widehat{\R}^{d}\setminus\{0\}, \quad \widehat a(k) \to 0, \; |k| \to
\infty,
\end{equation}
and our assumption
$$
a(z)\sim\frac{c_{0}}{|z|^{d+\alpha}}, \quad 0<\alpha\leq  2, \quad
\mbox{as} \quad z \to \infty
$$
implies that
$$
\widehat a(k) =  \int_{\R^d} e^{i(k,z)} a(z) dz = 1 - c_0 |k|^\alpha +
o(|k|^\alpha), \quad c_0>0, \quad  0< \alpha \le 2, \quad |k| \to 0.
$$

Each of the operators $H_k$ is essentially self-adjoint operator in
$L^2 (\R_+)$, and for small enough $k, \; |k| \le \delta$, has a
simple ground state $\psi_k(s) > 0$ analytically depending on the
perturbation operator $B_k$ which is a bounded operator in
$L^2(\R_+)$. The corresponding eigenvalue $\lambda_k$ is strictly
greater than $\lambda_0$, since $b(s)(1 - \widehat a(k)) \ge 0$. In this
situation one can use the standard Schr\"odinger perturbation
theory, see e.\,g. \cite{RS}, for the simple eigenvalue of the
perturbed operator $ H_k = H  + B_k$. Then in the case when $k$
is small enough: $|k| \le \delta$, the lowest eigenvalue $\lambda_k
$ of $H_k$ and the corresponding eigenfunction $\psi_k (s)$ have the
following representations:
\begin{equation}\label{C1}
\lambda_k = \lambda_0 + (B_k \psi_0, \psi_0) + O(||B_{k}||^{2})
\ = \ \lambda_0 + (1 - \widehat a(k))(b \psi_0, \psi_0 ) +
o(|k|^{\alpha}),
\end{equation}
and
\begin{equation}\label{C2}
\psi_k (s) = \psi_0 (s) + (H - \lambda_0)^{-1} ((B_k \psi_0,
\psi_0) \psi_0 - B_k \psi_0 ) + o(|k|^{\alpha}),
\end{equation}
where $\psi_0$ is the normalized eigenfunction of the operator $H$:
$(\psi_0, \psi_0) =1$, and the operator $(H - \lambda_0)^{-1}$ is
bounded in the invariant subspace $\psi_0^{\bot} = L^2(\R_+)\ominus
\{ \psi_0 \} $.
\\

Consider the parabolic problem associated with operator $-H_k$ in
$L^2(\R_+)$:
\begin{align*}
\frac{\partial \widehat\psi}{\partial t}  &=  - H_k \widehat \psi  =
\frac{\partial^2 \widehat \psi}{\partial s^2} - (v(s)  + b(s) (1 - \widehat
a(k))) \widehat \psi, \\
\widehat \psi &= \widehat \psi_t (s,k), \; s\in
\R_+,\; k \in \widehat\R^d,
\\
\widehat \psi_0 (s,k) &= \widehat u_0 (s,k) \in \ L^2(\R_+), \; k \in \widehat\R^d.
\end{align*}
Using the spectral decomposition for $\widehat \psi_t (s, k)$:
$$
\widehat \psi_t (s, k) = e^{-t H_k} \widehat u_0 (s,k) = \int_{\sigma (H_k)} e^{-t w} d E_{H_k} (w) \widehat u_0 (s,k)
$$
we get:
\begin{enumerate}
\item from the continuity arguments and condition (\ref{A1}) that for
any $\varepsilon > 0$ one can find $\delta = \delta (\varepsilon)>0$
such that when $ |k| \ge \delta$:
$$ \| \widehat \psi_t (s, k)
\|_{L^2(\R_+)}  \ \le \ e^{- t (\lambda_0 + \varepsilon)}  \| \widehat
u_0 (s,k) \|_{L^2(\R_+)},
$$
\item if $|k| \le \delta$, then
$$
\widehat \psi_t (s, k) = e^{-t \lambda_k} \psi_k(s) (\psi_k(s), \widehat
u_0 (s,k)) \ + \ \widehat \phi_t (s,k),
$$
where
$$ \| \widehat\phi_t (s, k) \|_{L^2(\R_+)}  \ \le \ e^{- t (\lambda_0 + \varepsilon)}  \| \widehat
u_0 (s,k) \|_{L^2(\R_+)}.
$$
\end{enumerate}



This implies
\begin{align}
&\quad\psi_t (s, x)  =   \frac{1}{(2 \pi)^d} \int_{\widehat \R^d} e^{-i
(k,x)}\widehat \psi_t (s,k) dk  \notag\\&=
\frac{1}{(2 \pi)^d} \int_{\{ |k| \le \delta \}} e^{-i (k,x)}
e^{-\lambda_k t} \psi_k (s) (\psi_k (s), \widehat u_0 (s,k)) dk \notag
\\&\quad+
\frac{1}{(2 \pi)^d} \int_{\{ |k| \le \delta \}} e^{-i (k,x)} \widehat
\phi_t (s,k) dk + \ \frac{1}{(2 \pi)^d} \int_{\{ |k| \ge \delta
\}} e^{-i (k,x)} \widehat \psi_t (s,k) dk \notag\\&=\label{D}
\frac{1}{(2 \pi)^d} \int_{\{ |k| \le \delta \}} e^{-i (k,x) -
\lambda_k t} \psi_k (s) (\psi_k (s), \widehat u_0 (s,k)) dk  +
\Psi_{\varepsilon}(t,s,x),
\end{align}
with
$$ \| \Psi_\varepsilon (t, s, x) \|_{L^2(\R_+, \R^d)}  \ \le \ e^{- t (\lambda_0 + \varepsilon)}  \|
u_0 (s,x) \|_{L^2(\R_+, \R^d)}.
$$

Finally we will find the asymptotic of the integral in (\ref{D}).
Using decompositions (\ref{C1})--(\ref{C2}) and after the change of
variables $ q = k (t \tilde c_0)^{1/\alpha}$ the integral in
(\ref{D}) can be written as follows
\begin{align}
&\frac{e^{-\lambda_0 t}}{(2 \pi)^d} \int_{\{ |k| \le \delta \}} e^{-i
(k,x)} e^{- \tilde c_0 |k|^\alpha t} \psi_0 (s) (\psi_0 (s), \widehat
u_0 (s,0)) dk \ \left( 1+ O \left(\frac{1}{t^{{1}/{\alpha}}}
\right) \right) \notag
\\[2mm]=
&\frac{e^{-\lambda_0 t}}{t^{d/\alpha}} \frac{C_d (\alpha)}{(2 \pi)^d}
\  \psi_0 (s) (\psi_0 (s), \widehat u_0 (s,0))\notag\\&\qquad\qquad\qquad\times \int_{\R^d} e^{-i \left(
q,\frac{x}{ (\tilde c_0 t)^{{1}/{\alpha}}} \right) - |q|^\alpha}
d q \ \left( 1+ O \left(\frac{1}{t^{{1}/{\alpha}}} \right)
\right),\label{Dasymp}
\end{align}
where the integral in (\ref{Dasymp})
$$
\frac{1}{(2 \pi)^d}  \int_{\R^d} e^{-i \left( q,\frac{x}{ (\tilde
c_0 t)^{{1}/{\alpha}}} \right) - |q|^\alpha} d q = p_\alpha
\left(  \frac{x}{ (\tilde c_0 t)^{{1}/{\alpha}}} \right)
$$
represents a density of a d-dimensional symmetric $\alpha$-stable
distribution, see e.\,g. \cite{Fel}, and it is an integer
function of $\dfrac{x}{t^{{1}/{\alpha}}}$.

Note that $\psi_0(s)
> 0, \ s \in (0, \infty)$ and $\widehat u_0 (s,k) > 0$ for small enough
$k$ (due to positivity of $u_0 (s,x))$.

Theorem is proved.
\end{proof}


\begin{thebibliography}{1}

\bibitem{Baake} Baake E. and Gabriel W., Biological evolution through
mutation, selection and drift: an introduction review, {\em Ann.
Rev. Comp. Phys.} \textsc{VII}, 203--264 (2000).

\bibitem{BSh} Berezin F.A. and Shubin M.A., {\em The Schr\"odinger equation},
(Moscow State University Publ., 1983) (Russian), (Kluwer, 1991)
(English).

\bibitem{Burger} Burger R., {\em The mathematical theory of selection,
recombination and mutation}, (NY: Wiley, 2000).

\bibitem{CK} Crow J.F. and Kimura M., The theory of genetic loads, {\em
In:} Proc. XI Int. Congr. Genetics, vol.2, 495--505, (Oxford:
Pergamon Press, 1964).

\bibitem{EN}
Engel K.-J. and Nagel R., {\em One-parameter semigroups for linear
  evolution equations\/}, vol. 194 of {\em Graduate Texts in
  Mathematics\/}, (Springer-Verlag, 2000).

\bibitem{Fel}
Feller W., {\em An introduction to probability theory and its applications}, vol. II, (John Wiley~\& Sons, Inc., New York--London--Sydney, 1966).

\bibitem{Ka} Kallenberg O., {\em Random Measures}, 4th edition,
(Academic Press; Akademie-Verlag, Berlin 1986).

\bibitem{K} Kimura M., A stochastic model concerning the maintenance
of genetic variability in quantitative characters, {\em Proc. Natl.
Acad. Sci. USA} \textbf{54}, 731--736 (1965).

\bibitem{KKP}
Kondratiev Y., Kutoviy O., and Pirogov S., Correlation functions and
invariant measures in continuous contact model.
\newblock {\em Infin. Dimens. Anal. Quantum Probab. Relat. Top.} \textbf{11}(2),
  231--258 (2008).

  \bibitem{RS}
Reed~M. and Simon~B.,
\newblock {\em Methods of Modern Mathematical Physics}, Vol. 4:
 Analysis of Operators.
\newblock Academic Press, 1978.


\bibitem{RSN}
Riesz F. and Sz\"{o}kefalvi-Nagy B., {\em Functional Analysis}, (NY:
Dover 1990).
\end{thebibliography}
\end{document}